\documentclass[letterpaper, 10 pt, conference]{ieeeconf}
\IEEEoverridecommandlockouts 
\usepackage{graphicx}
\usepackage{caption}
\usepackage[font={small,it}]{caption}
\usepackage{subcaption}
\usepackage{mathtools}
\usepackage{amsmath,amssymb,amsfonts}
\usepackage{xcolor}
\usepackage{cite}
\renewcommand{\baselinestretch}{0.997}
\graphicspath{{./figures/}}
\newtheorem{lemma}{Lemma} % Define a new lemma environment
\allowdisplaybreaks
\newcommand{\iu}{{i\mkern1mu}}

\title{\LARGE \bf
Probabilistically Robust Trajectory Planning of Multiple Aerial Agents
}

\author{Christian Vitale, 
Savvas Papaioannou, Panayiotis Kolios, and~Georgios Ellinas% <-this % stops a space
\thanks{This work was supported by the European Union's Horizon 2020 research and innovation programme under grant agreement No 739551 (KIOS CoE - TEAMING) and from the Republic of Cyprus through the Deputy Ministry of Research, Innovation and Digital Policy.}% <-this % stops a space
\thanks{C. Vitale, S. Papaioannou, and~G. Ellinas are with the KIOS Research and Innovation Center of Excellence and the Department of Electrical and Computer Engineering, University of Cyprus, Nicosia 1678, Cyprus. P. Kolios is with the KIOS Research and Innovation Center of Excellence and the Department of Computer Science, University of Cyprus, Nicosia 1678, Cyprus.}
\thanks{(e-mail:{\tt\small \{vitale.christian, papaioannou.savvas, pkolios, gellinas\}@ucy.ac.cy})}}

\begin{document}

\maketitle

\begin{abstract}
Current research on robust trajectory planning for autonomous agents aims to mitigate uncertainties arising from disturbances and modeling errors while ensuring guaranteed safety. Existing methods primarily utilize stochastic optimal control techniques with chance constraints to maintain a minimum distance among agents with a guaranteed probability. However, these approaches face challenges, such as the use of simplifying assumptions that result in linear system models or Gaussian disturbances, which limit their practicality in complex realistic scenarios.

To address these limitations, this work introduces a novel probabilistically robust distributed controller enabling autonomous agents to plan safe trajectories, even under non-Gaussian uncertainty and nonlinear systems. Leveraging exact uncertainty propagation techniques based on mixed-trigonometric-polynomial moment propagation, this method transforms non-Gaussian chance constraints into deterministic ones, seamlessly integrating them into a distributed model predictive control framework solvable with standard optimization tools. Simulation results demonstrate the effectiveness of this technique, highlighting its ability to consistently handle various types of uncertainty, ensuring robust and accurate path planning in complex scenarios.
\end{abstract}

\section{Introduction}
Trajectory planning for both single and multiple autonomous agents, encompassing ground robots \cite{Dixit2018,Madridano2021} and unmanned aerial vehicles (UAVs) \cite{Allaire2019}, is a critical research area in robotics, automation, and control. This significance stems from its importance across numerous application domains. These include searching and tracking of survivors \cite{Papaioannou2023Cyb} or moving targets in general \cite{9636065}, efficient resource deployment \cite{Kashino2020hybrid}, intelligent information collection for situational awareness \cite{schlotfeldt2018anytime}, exploration of unknown environments \cite{reinhart2020learning}, and automating inspection missions \cite{Papaioannou2022Icuas} for monitoring critical infrastructures.

The field of trajectory planning has matured significantly over the past few decades thanks to various methods including optimization-based and sampling-based methods \cite{Madridano2021,Dixit2018},
%,Souissi2013}
and more recently learning-based approaches \cite{Wang2021}. Nevertheless, there remains no dominant solution capable of robustly handling the inherent uncertainty \cite{Mesbah2016} in the agent dynamics and the environment. This limitation arises from the linear-Gaussian assumptions commonly used in most techniques, which struggle to generalize to real-world scenarios.

For instance, in \cite{vitale2022autonomous} a mathematical programming formulation is proposed for planning safe UAV trajectories while maximizing air traffic throughput, whereas in \cite{axelrod2018provably}, the authors utilize the rapidly exploring random tree (RRT) planning algorithm to design $\epsilon$-safe probabilistic trajectories. However, both techniques are based on linear-Gaussian assumptions for the system dynamics and uncertainty. Moreover, the work in \cite{dai2019chance} utilizes a quadrature-based sampling method combined with linear-quadratic-Gaussian (LQG) control to design feasible trajectories that satisfy risk bounds and meet certain optimality criteria. In \cite{Blackmore2010}, a robust particle-based stochastic model predictive control (MPC) framework is proposed with similar linear-Gaussian assumptions. 

Trajectory planning approaches for non-linear systems are investigated in \cite{Schwarting2017}, \cite{Papaioannou2023unscented},  and \cite{lew2020chance}. In \cite{Schwarting2017}, a non-linear model predictive controller (NMPC) is devised for the generation of safe driving trajectories in a shared-control framework, by taking into account the uncertainty in the prediction of other vehicles, which has to be of a Gaussian form. In \cite{Papaioannou2023unscented}, the authors utilize the unscented transform in order to plan robust coverage trajectories with a UAV exhibiting non-linear dynamics perturbed by Gaussian noise. Finally, \cite{lew2020chance} proposes a chance-constrained trajectory planning approach using sequential convex programming, accommodating non-convex constraints and Gaussian disturbances.

On the other hand, the work in \cite{summers2018distributionally}, proposes a distributionally robust sampling-based trajectory planning method which can handle non-Gaussian densities. Similarly, in \cite{paulson2020stochastic} a stochastic MPC approach is proposed which accounts for non-Gaussian stochastic disturbances. However, both works are restricted to systems exhibiting linear dynamics. To handle non-Gaussian uncertainty, a scenario-based approach is proposed in \cite{Calafiore2006}, but only for problems that can be formulated with convex constraints. Further, the authors in \cite{wang2020non}, and \cite{jasour2019risk} utilize moment-based methodologies for planning safe trajectories based on chance-constraints, without making any Gaussian and convexity assumptions. However, they do not consider system dynamics, and  instead, it is assumed that a behavior prediction system provides the distribution of future positions for the agent over the planning horizon. 

The contribution of this work is to address the existing gap in the state of the art, i.e., the design of a distributed robust controller that allows a team of autonomous aerial agents to cooperatively compute trajectories with \textit{safety guarantees}, even in the presence of nonlinear dynamics perturbed by non-Gaussian noise. Specifically:

\begin{itemize}
\item The work addresses safe trajectory planning for multiple autonomous agents with nonlinear dynamics perturbed by non-Gaussian noise, formulating it as a stochastic MPC problem.
\item The work demonstrates how integrating mixed-trigonometric-polynomial moment propagation into the proposed controller transforms this into a deterministic optimization problem within the moments' space. Subsequently, it shows how a distributed controller can define multi-agent trajectories that maintain safe distances with a desired probability.
\item The proposed approach is implemented using off-the-shelf optimization tools and is demonstrated through extensive simulation experiments.
\end{itemize}

The remaining sections of the paper are structured as follows. In Section \ref{sec:system_model}, the behavior of the aerial agents under nonlinear dynamics and non-Gaussian perturbations is discussed. Section \ref{sec:optimization_formulation} outlines the primary goal of solving the multi-agent trajectory planning problem while ensuring probabilistic safety robustness through chance constraints. Section \ref{sec:user_optimum} presents the proposed solution, starting with the characterization of the moments of the aerial agent system state. This is followed by a distributed deterministic optimization within the agents' moment space, reflecting the initial formulation with chance constraints. Section \ref{sec:performance_evaluation} evaluates the proposed solution's efficacy across various scenarios. Finally, Section \ref{sec:conclusion} concludes the study and proposes future research directions.

\section{System Model} \label{sec:system_model}

This study assumes that a specific number of aerial agents, i.e., UAVs, denoted as $M$ and represented by the set $\mathcal{M}$ (where $M=|\mathcal{M}|$ and $|.|$ indicates the set cardinality), operates collaboratively inside a shared 3D space, with the common goal of reaching their individually assigned destinations. The motion of UAV $\mu\in\mathcal{M}$ follows the discrete-time non-linear dynamical model shown in Eq. (\ref{eq:agent_dynamics}):
\vspace{-1pt}
\begin{equation} \label{eq:agent_dynamics}
\mathbf{x}_k =
\begin{bmatrix}
  x_k\\
  y_k\\
  z_k\\
  \psi_k
\end{bmatrix}
 = 
\begin{bmatrix}
  x_{k-1} + \Delta_s(u^v_{k-1}+\omega^v_{k-1})\cos\psi_{k-1}\\
  y_{k-1} + \Delta_s(u^v_{k-1}+\omega^v_{k-1})\sin\psi_{k-1}\\
  z_{k-1} + \Delta_s(u^z_{k-1}+\omega^z_{k-1})\\
  \psi_{k-1} + \Delta_s(u^\psi_{k-1}+\omega^\psi_{k-1})\\
\end{bmatrix}
\end{equation}
\vspace{-1pt}
where the index $\mu$ has been dropped to simplify the notation, $\mathbf{x}_k$ is the state of UAV $\mu$ at time step $k$, $x_k,y_k,z_k$ are its position in 3D Cartesian coordinates, $\psi_k$ is the orientation of the UAV on the horizontal plane, and $\Delta_s$ is the sampling interval. The control vector $\text{u}_k = [u^v, u^z, u^\psi]_k \in \mathbb{R}^3$, allows the UAV $\mu$ to change its speed ($u^v$), altitude ($u^z$), and direction ($u^\psi$), at time step k. Vector $\omega_k = [\omega^v, \omega^z, \omega^\psi]_k \in \mathbb{R}^3$ denotes the random disturbance acting on the system's control inputs. Specifically,  $\omega^v$, $\omega^z$, and $\omega^\psi$, denote the disturbance acting on $u^v$ (i.e., the horizontal speed control), $u^z$ (i.e., the vertical speed control), and $u^\psi$ (i.e., the direction control), respectively. In this study, it is assumed that $\omega^v$, $\omega^z$, and $\omega^\psi$ are mutually independent. Specifically, $\omega^v$ follows a Beta distribution ($\alpha = 1$, $\beta = 3$), $\omega^\psi$ follows a Uniform distribution in the range $[-0.1,0.1]$ rad/s, and $\omega^z$ follows a Gaussian distribution ($\mathcal{N}(0,0.3)$ m/s). However, it is important to note that both the dynamical model in Eq. (\ref{eq:agent_dynamics}) and control perturbations are selected for illustrative purposes and can be adjusted as needed.

With a slight change of variables, the system state $\mathbf{x}_k$ can be expressed recursively. To achieve this, instead of using the orientation on the horizontal plane $\psi_k$, $\mathbf{x}_k$ is adjusted to include $\cos\psi_k$ and $\sin\psi_k$. By leveraging basic trigonometric identities, $\cos\psi_k$ and $\sin\psi_k$ can be expressed as follows:

\begin{equation}
\begin{aligned}
    \cos\psi_k &=  \cos\psi_{k-1}\cos(\Delta_s(u^\psi_{k-1}+\omega^\psi_{k-1}))\\
    &\quad -\sin\psi_{k-1}\sin(\Delta_s(u^\psi_{k-1}+\omega^\psi_{k-1})) \\
	\sin\psi_k &=  \sin\psi_{k-1}\cos(\Delta_s(u^\psi_{k-1}+\omega^\psi_{k-1}))\\
	&\quad +\cos\psi_{k-1}\sin(\Delta_s(u^\psi_{k-1}+\omega^\psi_{k-1})).
\end{aligned}
\end{equation}

Following the presented transformation, the final system state $\mathbf{x}_k$ of UAV $\mu$ can be expressed as:

\begin{equation} \label{eq:agent_dynamics_augmented}
\mathbf{x}_k =
\begin{bmatrix}
  x_k\\
  y_k\\
  z_k\\
  \cos\psi_k\\
  \sin\psi_k\\
\end{bmatrix}
 = 
A_{k-1} \mathbf{x}_{k-1} + b_{k-1},
\end{equation}
where $A_k$ and $b_k$ are further given by:
\begin{equation}
A_k =
\begin{bmatrix}
  	1 & 0 & 0 & \Delta_s(u^v_k\mathord{+}\omega^v_k) & 0 \\
    0 & 1 & 0 & 0 & \Delta_s(u^v_k\mathord{+}\omega^v_k) \\
    0 & 0 & 1 & 0 & 0 \\
    0 & 0 & 0 & \cos(\Delta_s(u^\psi_k\mathord{+}\omega^\psi_k)) & -\sin(\Delta_s(u^\psi_k\mathord{+}\omega^\psi_k)) \\
    0 & 0 & 0 & \sin(\Delta_s(u^\psi_k\mathord{+}\omega^\psi_k)) & \cos(\Delta_s(u^\psi_k\mathord{+}\omega^\psi_k)) \\
\end{bmatrix}
\end{equation}
\begin{equation}
b_k =
\begin{bmatrix}
  0\\
  0\\
  \Delta_s(u^z_k+\omega^z_k)\\
  0\\
  0\\
\end{bmatrix}
\end{equation}

Assuming independent perturbations across time steps, the UAV dynamics follow the Markov property, i.e, the UAV's state in the next time step is solely determined by its current state and control input.

\section{Problem Formulation}
\label{sec:optimization_formulation}

The problem tackled in this work can be stated as follows: \textit{Given a set of UAVs $\mathcal{M}$ and sufficiently large planning horizon $T$, find the joint mobility, i.e., $\text{u}^\mu_{[0,...,T-1]}$ for $\mu\in\mathcal{M}$, that allows each UAV to reach its destination $(c^\mu_x,c^\mu_y,c^\mu_z) \in \mathbb{R}^3$ safely, i.e., respecting a minimum distance $\text{d}_{\text{MIN}}$ with all the other UAVs with probability $1-\epsilon$.}

To solve this problem, we first devise a robust centralized controller, i.e., as an Air Traffic Manager (ATM), utilizing MPC to safely and efficiently guide a team of $M$ UAVs to their destinations. By centralizing decision-making, the ATM treats the system holistically, synchronizing UAV movements and optimizing paths collectively. 
%Then, in the following section, we illustrate the development of a robust distributed controller for the same problem, incorporating precise uncertainty propagation techniques via trigonometric-polynomial moment propagation. This technique convert probability constraints into deterministic safety constraints, accommodating both non-Gaussian and unbounded disturbances.

%In the next section, we assume that the $M$ UAVs collaboratively plan their flight paths in the aerial environment, aiming for smooth and efficient travel to their destination in the shortest time possible, while avoiding collisions. Typically, a centralized controller, i.e., an Air Traffic Manager (ATM), is employed to tackle this problem. By centralizing the decision-making process, the ATM considers the system as a whole, synchronizes the UAVs' movements, and optimizes their paths collectively. This centralized coordination, enabled by the holistic view of the airspace available to the ATM, accounts for all UAV positions and velocities simultaneously.
%Consequently, our work first introduces a centralized method based on receding horizon and MPC, ensuring not only individual UAV efficiency but also the overall optimization and safety of the entire UAV fleet.

In the centralized approach, we assume perfect connectivity between UAVs and the ATM. At each time step $k$, all UAVs share their system state estimation $\mathbf{x}_k$ with the ATM. Leveraging this comprehensive view of the 3D flying environment, the ATM jointly determines the trajectories of the UAVs. It calculates the controls $\text{u}^\mu_{[0,...,T-1]}$ for $\mu\in\mathcal{M}$ and a sufficiently large planning horizon $T$, ensuring: (i) the minimization of the expected squared distance of UAVs to their destination at the end of the planning horizon $T$, and (ii) the minimization of the absolute value of controls applied by the UAVs. This objective function ensures a gradual approach to the destination avoiding abrupt maneuvers. Additionally, during control selection, the ATM guarantees the required level of safety among UAVs.

Once the controls for all UAVs are determined, the ATM broadcasts them to the UAVs. Consequently, each UAV applies the specified controls for time step $k$, i.e., $[u^v, u^z, u^\psi]_0$, and estimates its updated system state $\mathbf{x}_{k+1}$. Upon receiving the updated estimation of the UAVs' system state, the ATM refines control decisions for the subsequent planning horizon.

Formally, the objective function utilized by the ATM at time slot $k$ can be expressed as:\vspace{-3.5mm}

\begin{equation}\label{eq:objective_function}
\begin{aligned}
J_{[\text{u}^1_k,...,\text{u}^\mu_k,...,\text{u}^M_k]} &= \sum_{\mu=1}^M (\mathbb{E}[(x^\mu_T\mathord{-}c^\mu_x)^2\mathord{+}(y^\mu_T\mathord{-}c^\mu_y)^2\mathord{+}(z^\mu_T\mathord{-}c^\mu_z)^2]  \\[-1mm]
&\quad\quad \mathord{+} w \sum_{\tau=0}^{T-1} |u^{v,\mu}_\tau| + |u^{z,\mu}_\tau| + |u^{\psi,\mu}_\tau| ).
\end{aligned}
\end{equation}

Here, $\mathbb{E}[\cdot]$ denotes the expectation of UAV $\mu$'s quadratic distance from its destination at time step $T$. Given the system state description in Eq.~(\ref{eq:agent_dynamics_augmented}), the expectation of the UAVs' location at the end of the planning horizon $T$ is uniquely determined by the applied controls and by the control noise distributions. Lastly, the parameter $w$ functions as a weighting factor, balancing smoothness with the speed of approaching the destination. As shown in \cite{shanno1971linear}, the absolute control values in the objective function can be replaced by slack variables and additional linear constraints.

For what concerns safety, the ATM selects the controls of the UAVs to ensure that the probability of two UAVs being closer than a safe distance $\text{d}_{\text{MIN}}$ to each other at any point of the planning horizon is bounded. That is, for all $\tau\in [1,...,T]$ and for all $i\neq \mu\in \mathcal{M}$, the following chance constraints must be respected:\vspace{-3.5mm}

\begin{equation}\label{eq:safe_distance}
P((x^\mu_\tau-x^i_\tau)^2+(y^\mu_\tau-y^i_\tau)^2+(z^\mu_\tau-z^i_\tau)^2\leq \text{d}^2_{\text{MIN}})\leq\epsilon.
\end{equation}

The introduced chance constraints highlight the intricate nature of the trajectory planning problem at hand, where probabilistically respecting the distance between two UAVs depends on both their locations and, consequently, both their controls. Summarizing, in each time step $k$, the ATM solves the following Chance-Constrained MPC (\texttt{ATM CC-MPC}) optimization where $\tau\in [1,...,T]$ and $\mu\in\mathcal{M}$:

\begin{subequations}
\begin{alignat}{3}
&\rlap{\textbf{Problem }\texttt{ATM CC-MPC}:} & & & \nonumber\\[-1mm]
& \min\limits_{[\text{u}^1_k,...,\text{u}^\mu_k,...,\text{u}^M_k]} \quad J_{[\text{u}^1_k,...,\text{u}^\mu_k,...,\text{u}^M_k]} & & &\label{eq:obj_user}\\[1mm]
&\text{subject to:} &  &  & \nonumber\\[-1mm]
& \mathbf{x}^\mu_\tau = A^\mu_{\tau-1} \mathbf{x}^\mu_{\tau-1} + b^\mu_{\tau-1} &\quad\quad\quad &\forall \tau, \forall \mu\label{eq:eq1_user}\\[-1mm]
& \text{u}^v_\text{MIN} \leq \text{u}^{v,\mu}_\tau \leq \text{u}^v_\text{MAX} & \quad & \forall \tau, \forall \mu  &\label{eq:eq2_user}\\[-1mm]
& \text{u}^z_\text{MIN} \leq \text{u}^{z,\mu}_\tau \leq \text{u}^z_\text{MAX} & \quad & \forall \tau, \forall \mu &\label{eq:eq3_user}\\[-1mm]
& \text{u}^\psi_\text{MIN} \leq \text{u}^{\psi,\mu}_\tau \leq \text{u}^\psi_\text{MAX} & \quad & \forall \tau, \forall \mu &\label{eq:eq4_user}\\[-1mm]
& |\text{u}^{v,\mu}_\tau\mathord{-}\text{u}^{v,\mu}_{\tau-1}|\leq\Delta \text{u}^v &\quad\quad\quad & \forall \tau, \forall \mu & \label{eq:eq5_user}\\[-1mm]
& |\text{u}^{z,\mu}_\tau\mathord{-}\text{u}^{z,\mu}_{\tau-1}|\leq\Delta \text{u}^z &\quad\quad\quad & \forall \tau, \forall \mu & \label{eq:eq6_user}\\[-1mm]
& \rlap{$\displaystyle P((x^\mu_\tau-x^i_\tau)^2+(y^\mu_\tau-y^i_\tau)^2+(z^\mu_\tau-z^i_\tau)^2\leq \text{d}^2_{\text{MIN}})\leq\epsilon
 $} &&& \nonumber\\[-1mm]
& &\quad\quad\quad & \forall \tau, \forall \mu, \forall i\mathord{\neq} \mu &\label{eq:eq7_user}
\end{alignat}
\end{subequations}

\vspace{-0.5mm}
The \texttt{ATM CC-MPC} optimization updates all UAV controls at once. As described earlier, in Eq. (\ref{eq:obj_user}), the \texttt{ATM CC-MPC} framework is designed to obtain controls $[\text{u}^1_k,...,\text{u}^\mu_k,...,\text{u}^M_k]$ that collectively and smoothly minimize the expected distance of UAVs to their destination. The constraint in Eq. (\ref{eq:eq1_user}) ensures that the objective function accounts for the non-linear relationship between the controls $[\text{u}^1_k,...,\text{u}^\mu_k,...,\text{u}^M_k]$ and the UAVs' system state, as described in Sec. \ref{sec:system_model}.

The remaining constraints specify how controls can be selected. Equations (\ref{eq:eq2_user})-(\ref{eq:eq4_user}) require controls to fall within specified ranges. Equations (\ref{eq:eq5_user})-(\ref{eq:eq6_user}) limit the maximum deviation between consecutive controls, ensuring realistic control decisions. Lastly, Eq. (\ref{eq:eq7_user}) enforces that the probability of the distance between two UAVs being below $\text{d}_{\text{MIN}}$ at any time step $\tau$ in the planning horizon is at most $\epsilon$. This guarantees that the future trajectories of the UAVs adhere to strict safety standards.

In addition to the inherent complexity in determining the distribution of the distance between pairs of UAVs (both of which depend on the control disturbances distributions $[\omega^v, \omega^z, \omega^\psi]$), it is crucial to note the exponential growth in the number of chance constraints. Indeed, each pair of UAVs in the shared flying environment introduces a chance constraint for each time step. Consequently, the centralized approach, while theoretically optimal, becomes increasingly infeasible for a large number of UAVs. To address these challenges, next, a distributed MPC approximating in practice the optimal solutions of  \texttt{ATM CC-MPC} is introduced, incorporating exact uncertainty propagation techniques to efficiently implement the chance constraints in Eq. (\ref{eq:safe_distance}).

\section{A Distributed Multi-UAV Trajectory Planning Under Non-Gaussian Noise}
\label{sec:user_optimum}

Exploiting the same approach as in \cite{richards2007robust}, a distributed implementation is presented as an efficient approximation of the centralized \texttt{ATM CC-MPC} introduced in the previous section. Firstly, the objective function in Eq. (\ref{eq:objective_function}), which jointly optimizes smoothness and arrival to destination, can be easily decoupled for each UAV. This inherent separability allows individual UAVs to optimize their trajectories independently while also optimizing the global collective objective shown in Eq. (\ref{eq:objective_function}). Secondly, while the chance constraints are inherently coupled due to the interdependence of UAV positions, UAVs can ensure these constraints are met collectively even when they solve the problem distributively, as explained later.

Based on the above considerations, in this section, we introduce our proposed solution, the Distributed Deterministic MPC under Non-Gaussian Noise (\texttt{Distributed Det-MPC NG}). Section \ref{subsec:moments_computation} covers the computation of the first $n$ moments of the UAVs' system state. Section \ref{subsec:deterministic_bound_safety} demonstrates how these moments are utilized to derive a deterministic bound for the safety chance constraint in Eq. (\ref{eq:safe_distance}). Exploiting this result, Sec. \ref{subsec:proposed_MPC} demonstrates the transformation of the centralized optimization \texttt{ATM CC-MPC} in its distributed counterpart \texttt{Distributed Det-MPC NG}.

\subsection{UAV's System State Moments Computation}\label{subsec:moments_computation}
A methodology for calculating the initial $n$ moments of UAV $\mu$'s system state is presented herein. This method relies on findings from \cite{jasour2021moment}, where computations involving the expectations of mixed trigonometric-polynomial functions of random variables are detailed. The specific computation technique is outlined in the following lemma.

\begin{lemma}
    Let $\theta$ be a random variable with characteristic function $\Phi_\theta(t)$. Given $(\alpha_1,\alpha_2,\alpha_3) \in \mathbb{R}^3$ and a constant scaling factor $\delta$, the expectation of the mixed-trigonometric polynomial function $m_{\delta,\theta^{\alpha_1} c^{\alpha_2}_\theta s^{\alpha_3}_\theta}=\mathbb{E}[(\delta \theta)^{\alpha_1} (\cos(\delta \theta)^{\alpha_2}) (\sin(\delta \theta))^{\alpha_3}]$ is \cite{jasour2021moment}:

\vspace{-4mm}
\begin{equation}\label{eq:lemma_1}
\begin{aligned}
m_{\delta,\theta^{\alpha_1} c^{\alpha_2}_\theta s^{\alpha_3}_\theta}=\frac{1}{\iu^{(\alpha_1+\alpha_3)}2^{(\alpha_2+\alpha_3)}} \sum_{g=0}^{\alpha_2}\sum_{h=0}^{\alpha_3}\binom{\alpha_2}{g}\binom{\alpha_3}{h}\\
 (-1)^{(\alpha_3-h)}\frac{\partial^{\alpha_1}}{\partial t^{\alpha_1}}\Phi_\theta(\delta\hspace{0.5mm}t)|_{t=2(g+h)-\alpha_2-\alpha_3}.
\end{aligned}
\end{equation}

\end{lemma}

The recursive dynamics expressed in Eq. (\ref{eq:agent_dynamics_augmented}) shows that the first moment of the system state $\mathbf{x}_k$ at time step $k$, which is denoted hereinafter as $m^1_k$, can be computed as a linear summation of mixed-trigonometric functions of the controlling random noises $[\omega^v, \omega^z, \omega^\psi]_k$ and of the controls $[u^v, u^z, u^\psi]_k$. Indeed, let $m^1_k=\mathbb{E}[\mathbf{x}_k]=[\mathbb{E}[x_k],$ $\mathbb{E}[y_k],\mathbb{E}[z_k],\mathbb{E}[\cos\psi_k],\mathbb{E}[\sin\psi_k]]^\intercal$. Applying the expectation operator to Eq. (\ref{eq:agent_dynamics_augmented}) with constant control variables, then:

\vspace{-1mm}
\begin{equation}\label{eq:first_moment}
m^1_k = A^{mom_1}_{k-1} m^1_{k-1} + b^{mom_1}_{k-1},
\end{equation}
where $A^{mom_1}_{k-1}$ and $b^{mom_1}_{k-1}$ can be computed directly by applying Lemma 1. Specifically:
\begin{equation}
A^{mom_1}_k =
\begin{bmatrix}
  	1 & 0 & 0 & \Delta_s u^v_k + m_{\Delta_s,\omega^v} & 0 \\
    0 & 1 & 0 & 0 & \Delta_s u^v_k + m_{\Delta_s,\omega^v} \\
    0 & 0 & 1 & 0 & 0 \\
    0 & 0 & 0 & \gamma_k & -\lambda_k \\
    0 & 0 & 0 & \lambda_k & \gamma_k \\
\end{bmatrix}
\end{equation}
\begin{equation}
b^{mom_1}_k =
\begin{bmatrix}
  0\\
  0\\
  \Delta_s u^z_k + m_{\Delta_s,\omega^z}\\
  0\\
  0\\
\end{bmatrix},
\end{equation}
with:
\begin{align*}
\gamma_k &= \cos(\Delta_s u^\psi_k)m_{\Delta_s,c_{\omega^\psi}}-\sin(\Delta_s u^\psi_k)m_{\Delta_s,s_{\omega^\psi}}\\
\lambda_k &= \sin(\Delta_s u^\psi_k)m_{\Delta_s,c_{\omega^\psi}}+\cos(\Delta_s u^\psi_k)m_{\Delta_s,s_{\omega^\psi}},
\end{align*}
and $m_{\Delta_s,\omega^v}$, $m_{\Delta_s,c_{\omega^\psi}}$, $m_{\Delta_s,s_{\omega^\psi}}$, and $m_{\Delta_s,\omega^z}$ computed as in Eq. (\ref{eq:lemma_1}).

The moments of higher order can be obtained following the same procedure, but they are not presented here for the sake of space. As an example, considering the second moment $m^2_k=[\mathbb{E}[x_k^2],\text{ }\mathbb{E}[x_k y_k],\text{ }\mathbb{E}[x_k z_k],\text{ }\mathbb{E}[x_k \cos \psi_k],\text{ }\mathbb{E}[x_k \sin \psi_k],\text{ }\mathbb{E}[y_k^2],\text{ }$ $\mathbb{E}[y_k z_k],\text{ }\mathbb{E}[y_k \cos \psi_k],\text{ }\mathbb{E}[y_k \sin \psi_k],\text{ }\mathbb{E}[z_k^2],\text{ }\mathbb{E}[z_k \cos \psi_k],\text{ }$ $\mathbb{E}[z_k \sin \psi_k],\text{ }\mathbb{E}[\cos^2 \psi_k],\text{ }\mathbb{E}[\cos \psi_k \sin \psi_k],\text{ }\mathbb{E}[\sin^2\psi_k]]^\intercal$, and applying the expectation operator to each element in $m^2_k$ as represented in Eq. (\ref{eq:agent_dynamics_augmented}), the second moment can be written as:

\begin{equation}\label{eq:second_moment}
m^2_k = A^{mom_2}_{k-1} m^2_{k-1} + b^{mom_2}_{k-1},
\end{equation}
where $A^{mom_2}_k$ and $b^{mom_2}_k$ are, respectively, a $15$x$15$ matrix and a $15$x$1$ vector obtained exploiting Lemma 1. For instance, if the fourth element of $m^2_k$ is computed, i.e., $\mathbb{E}[x_k \cos\psi_k]$, then:

\begin{equation}\label{eq:second_moment_sample}
\begin{aligned}
	\mathbb{E}[x_k \cos\psi_k] = \mathbb{E}&[(x_{k-1} \mathord{+} \Delta_s(u^v_{k-1}\mathord{+}\omega^v_{k-1})\cos\psi_{k-1})\cdot\\
    &\quad(\cos\psi_{k-1}\cos(\Delta_s(u^\psi_{k-1}\mathord{+}\omega^\psi_{k-1}))\mathord{-}\\
    &\quad\sin\psi_{k-1}\sin(\Delta_s(u^\psi_{k-1}\mathord{+}\omega^\psi_{k-1})))]\\
    = \gamma_{k-1}\mathbb{E}&[x_{k-1} \cos\psi_{k-1}]\mathord{-}\lambda_{k-1}\mathbb{E}[x_{k-1} \sin\psi_{k-1}]+\\
     \gamma_{k-1}&(\Delta_s u^v_{k-1} \mathord{+} m_{\Delta_s,\omega^v})\mathbb{E}[\cos^2\psi_{k-1}]-\\
    \lambda_{k-1}&(\Delta_s u^v_{k-1} \mathord{+} m_{\Delta_s,\omega^v})\mathbb{E}[\cos\psi_{k-1}\sin\psi_{k-1}],\\
\end{aligned}
\end{equation}
which can be easily used to obtain the corresponding rows of $A^{mom_2}_k$ and $b^{mom_2}_k$. Overall, assuming that the distributions of the existing perturbations have finite moments, the moments of the UAV's system state can be expressed in closed form solely depending on a nonlinear function of the applied controls $\text{u} = [u^v, u^z, u^\psi]$, where the agent's initial state serve as input parameters. This implies that while computing the moments of the agent's states may be computationally intensive, this task needs to be executed only once, prior to the agent's journey start.
\vspace{-1mm}

\subsection{Deterministic Safety Distance Under Non-Gaussian Noise}
\label{subsec:deterministic_bound_safety}

Imposing safety guarantees, as expressed in Eq. (\ref{eq:safe_distance}), implies knowing the distribution of the unimodal random variable represented by the squared distance between UAVs $\mu$ and $i$ at time step $k$, i.e.:

\vspace{-4mm}
\begin{equation}\label{eq:obstacle_equation}
d^{\mu,i}_k = (x^\mu_k-x^i_k)^2+(y^\mu_k-y^i_k)^2+(z^\mu_k-z^i_k)^2.
\end{equation}

Given the complexity of this task, in the following derivations, an upperbound on the probability $P(d^{\mu,i}_k\leq \text{d}^2_{\text{MIN}})$ is computed. Then, imposing that such upperbound is lower than $\epsilon$ grants the desired stochastic safety guarantee. The upperbound on $P(d^{\mu,i}_k\leq \text{d}^2_{\text{MIN}})$ is obtained exploiting a known concentration inequality, i.e., the one-sided Vysochanskij–Petunin inequality, as showed in the following Lemma.

\begin{lemma}
    Let $f^{\mu,i}_k$ be a random variable representing the difference between the squared distance between UAVs $\mu$ and $i$ at time step $k$ and a constant $\text{d}^2_{\text{MIN}}$, i.e., $f^{\mu,i}_k=d^{\mu,i}_k-\text{d}^2_{\text{MIN}}$. Then, a valid upperbound on the probability $P(f^{\mu,i}_k\leq0)$ is:
    
\begin{equation}\label{eq:lemma_2}
P(f^{\mu,i}_k\leq0)\leq \frac{4}{9} \frac{\mathbb{E}[(f^{\mu,i}_k)^2]-\mathbb{E}[f^{\mu,i}_k]^2}{\mathbb{E}[f^{\mu,i}_k]^2},
\end{equation}

if: 
\begin{equation}
\begin{aligned}\label{eq:VP_1}
\mathbb{E}[f^{\mu,i}_k]&\geq 0,\\
\mathbb{E}[f^{\mu,i}_k]^2&\geq\frac{5}{8}\mathbb{E}[(f^{\mu,i}_k)^2].
\end{aligned}
\end{equation}
\end{lemma}

\begin{proof}
In \cite{jasourHW21}, concentration inequalities were demonstrated as effective tools for bounding chance constraints. Specifically, the one-sided Vysochanskij–Petunin inequality, applicable to unimodal random variables $X$ with $r\geq 0$ (where $r\in \mathbb{R}$), is defined as follows:

\vspace{-2mm}
\begin{equation}\label{eq:VP_inequality}
P(X-\mathbb{E}[X]\geq r)\leq \frac{4}{9}\frac{\mathbb{E}[(X-\mathbb{E}[X])^2]}{r^2+\mathbb{E}[(X-\mathbb{E}[X])^2]}.
\end{equation}
This inequality holds when the following condition holds: 
\begin{equation}\label{eq:VP_inequality_cond}
r^2\geq\frac{5}{3}\mathbb{E}[(X-\mathbb{E}[X])^2]
\end{equation}

Applying this inequality to the specific random unimodal variable $X=-f^{\mu,i}_k$, setting $r=\mathbb{E}[f^{\mu,i}_k]$, and utilizing the property $\mathbb{E}[(X-\mathbb{E}[X])^2]=\mathbb{E}[X^2]-\mathbb{E}[X]^2$, straightforward substitutions in Eq. (\ref{eq:VP_inequality}) yield Eq. (\ref{eq:lemma_2}). 

It is important to note that the one-sided Vysochanskij–Petunin inequality is applicable only under the conditions $r\geq 0$ and Eq. (\ref{eq:VP_inequality_cond}). Transforming these requirements for the specific random variable $X=-f^{\mu,i}_k$ and for $r=\mathbb{E}[f^{\mu,i}_k]$,  the first and second conditions in Eqs. (\ref{eq:VP_1}) are derived, respectively.
\end{proof}

Lemma 2 establishes the upperbound on $P(f^{\mu,i}_k\leq0)$, i.e., on $P(d^{\mu,i}_k\leq \text{d}^2_{\text{MIN}})$, as a function of $\mathbb{E}[f^{\mu,i}_k]$ and of $\mathbb{E}[(f^{\mu,i}_k)^2]$. By expanding $f^{\mu,i}_k$ and $(f^{\mu,i}_k)^2$, and applying the expectation operator while considering the statistical independence of the system states of UAVs $\mu$ and $i$, the upper bound can be expressed in terms of the first four moments of their system states. For instance, $\mathbb{E}[f^{\mu,i}_k]$ is computed as:

\vspace{-2mm}
\begin{equation}\label{eq:average_obstacle}
\begin{aligned}
\mathbb{E}&[(x^\mu_k-x^i_k)^2+(y^\mu_k-y^i_k)^2+(z^\mu_k-z^i_k)^2 - \text{d}^2_{\text{MIN}}] = \\
&\mathbb{E}[(x^\mu_k)^2]-2\mathbb{E}[x^\mu_k]\mathbb{E}[x^i_k]+\mathbb{E}[(x^i_k)^2]+\\[-1mm]
&\quad\mathbb{E}[(y^\mu_k)^2]-2\mathbb{E}[y^\mu_k]\mathbb{E}[y^i_k]+\mathbb{E}[(y^i_k)^2]+\\[-1mm]
&\quad\mathbb{E}[(z^\mu_k)^2]-2\mathbb{E}[z^\mu_k]\mathbb{E}[z^i_k]+\mathbb{E}[(z^i_k)^2]-\text{d}^2_{\text{MIN}},
\end{aligned}
\end{equation}
where the moments of UAV $\mu$ can be expressed as a non-linear function of the controls $\text{u}^\mu_{[0:T-1]}$, as shown in Sec. \ref{subsec:moments_computation}, and where the moments of UAV $i$ are known, as explained in the following section.

\subsection{The \texttt{Distributed Det-MPC NG} Framework}
\label{subsec:proposed_MPC}

In the \texttt{Distributed Det-MPC NG} framework, the $M$ UAVs in the aerial environment devise their flight paths employing a distributed approach based on receding horizon and MPC \cite{richards2007robust}. The set of UAVs, $\mathcal{M}$, is organized in a sorted manner, and UAVs determine their trajectories at each time step following this predetermined sequence. Perfect connectivity is assumed among UAVs. Within this ordered sequence, at each time step, UAV $1$ selects controls $\text{u}^1_{[0,...,T-1]}$ to determine its forthcoming trajectory, taking into consideration the previous plans of UAVs $2,...,M$. Once established, UAV $1$ promptly shares its plan with the rest of the UAVs. Specifically, after obtaining the controls $\text{u}^1_{[0:T-1]}$ and applying the concepts in Sec. \ref{subsec:moments_computation}, UAV $1$ broadcasts the four moments of its system state for all time steps in the planning horizon. In general, UAV $\mu$ in the sorted list selects controls $\text{u}^\mu_{[0,...,T-1]}$ by integrating the new plans of UAVs $1,...,\mu-1$ and the old plans of UAVs $\mu+1,...,M$. Again, as soon as obtained, UAV $\mu$ broadcasts the obtained planned plan, i.e., the first four moments, to the other UAVs.

In detail, UAV $\mu$ solves the following non-linear \texttt{Distributed Det-MPC NG} optimization to select its controls, where unless otherwise specified, $k\in [1,...,T]$, $j$ denotes any UAV, encompassing both UAV $\mu$ about to select its controls and any other UAV $i$ ($i\neq \mu$), and $n\in[1,...,4]$:

\vspace{-4mm}
\begin{subequations}
\begin{alignat}{3}
&\rlap{\textbf{Problem } \texttt{Distributed Det-MPC NG}:} & & & \nonumber\\[-1mm]
& \min\limits_{\text{u}^\mu_{[0:T-1]}} \quad \rlap{$\displaystyle \mathbb{E}[(x^\mu_T\mathord{-}c^\mu_x)^2\mathord{+}(y^\mu_T\mathord{-}c^\mu_y)^2\mathord{+}(z^\mu_T\mathord{-}c^\mu_z)^2]\mathord{+}$} && \nonumber\\[-4mm]
& \quad\quad\quad\quad\quad\quad\rlap{$\displaystyle  w \sum_{\tau=0}^{T-1} |u^{v,\mu}_\tau| + |u^{z,\mu}_\tau| + |u^{\psi,\mu}_\tau| $} && \label{eq:obj_user_f}\\[-2mm]
&\text{subject to:} &  &  & \nonumber\\[-1mm]
& m^{j,n}_k = A^{mom_{j,n}}_{k-1} m^{j,n}_{k-1} + b^{mom_{j,n}}_{k-1} &\quad &\forall k, \forall j, \forall n\label{eq:eq1_user_f}\\[-1mm]
& \text{u}^v_\text{MIN} \leq \text{u}^{\mu,v}_k \leq \text{u}^v_\text{MAX} & \quad & \forall k &\label{eq:eq2_user_f}\\[-1mm]
& \text{u}^z_\text{MIN} \leq \text{u}^{\mu,z}_k \leq \text{u}^z_\text{MAX} & \quad & \forall k &\label{eq:eq3_user_f}\\[-1mm]
& \text{u}^\psi_\text{MIN} \leq \text{u}^{\mu,\psi}_k \leq \text{u}^\psi_\text{MAX} & \quad & \forall k &\label{eq:eq4_user_f}\\[-1mm]
& |\text{u}^{\mu,v}_k\mathord{-}\text{u}^{\mu,v}_{k-1}|\leq\Delta \text{u}^v &\quad & \forall k & \label{eq:eq5_user_f}\\[-1mm]
& |\text{u}^{\mu,z}_k\mathord{-}\text{u}^{\mu,z}_{k-1}|\leq\Delta \text{u}^z &\quad & \forall k & \label{eq:eq6_user_f}\\[-1mm]
& \frac{4}{9} \frac{\mathbb{E}[(f^{\mu,i}_k)^2]-\mathbb{E}[f^{\mu,i}_k]^2}{\mathbb{E}[f^{\mu,i}_k]^2}\leq \epsilon & \quad & \forall i\mathord{\neq} \mu,\forall k &\label{eq:eq7_user_f}\\[-1mm]
& \mathbb{E}[f^{\mu,i}_k]\geq 0 & \quad & \forall i\mathord{\neq} \mu,\forall k &\label{eq:eq8_user_f}\\[-1mm]
& \mathbb{E}[f^{\mu,i}_k]^2\geq\frac{5}{8}\mathbb{E}[(f^{\mu,i}_k)^2]& \quad & \forall i\mathord{\neq} \mu,\forall k &\label{eq:eq9_user_f}
\end{alignat}
\end{subequations}

\begin{figure*}[ht]
    \centering
    \begin{subfigure}[b]{0.27\textwidth}
        \centering
        \includegraphics[width=\textwidth]{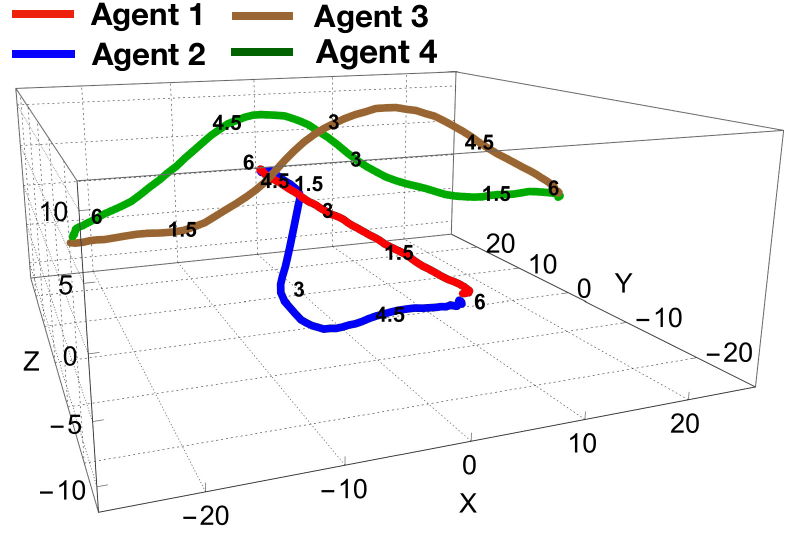}
        \caption{First View Point}
        \label{fig:trajectory1}
    \end{subfigure}%
    \hfill
    \begin{subfigure}[b]{0.29\textwidth}
        \centering
        \includegraphics[width=\textwidth]{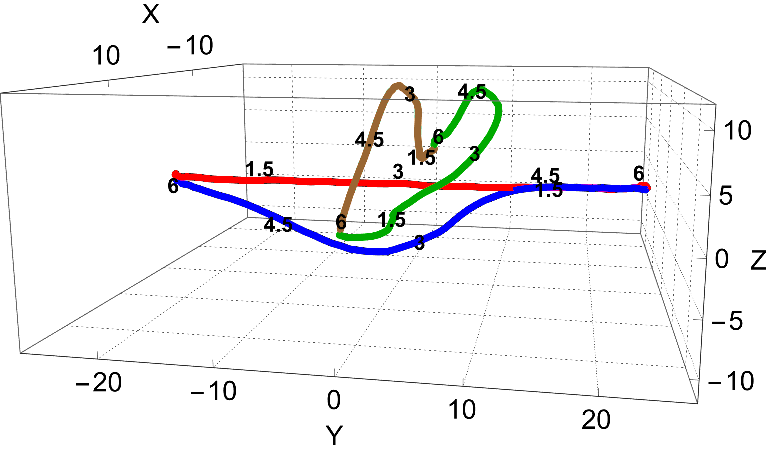}
        \caption{Second View Point}
        \label{fig:trajectory2}
    \end{subfigure}%
    \hfill
    \begin{subfigure}[b]{0.27\textwidth}
        \centering
        \includegraphics[width=\textwidth]{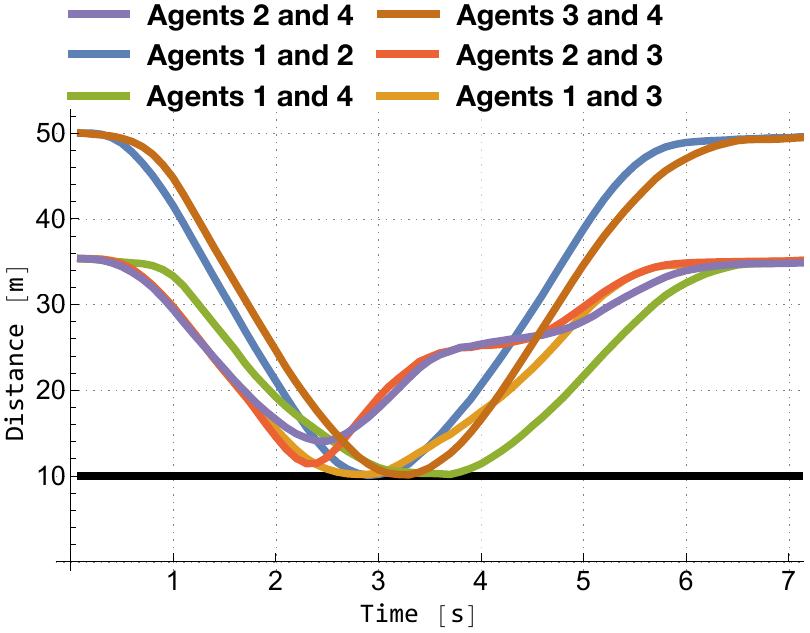}
        \caption{Pairwise Distances}
        \label{fig:pairwisedistance}
    \end{subfigure}
    \caption{(a, b) shows executed UAV trajectories respecting safety distance $\text{d}_\text{MIN}=10$ m with probability $1-\epsilon=0.9$ for non-Gaussian noise. Time marks along the trajectories indicate progression. (c) shows pairwise distances among UAVs that, in all time steps, respect safety distance $\text{d}_\text{MIN}=10$ m.}
        \vspace{-3mm}
    \label{fig:fulltrajectory}
\end{figure*}

In \texttt{Distributed Det-MPC NG}, the objective function in Eq. (\ref{eq:obj_user_f}) only accounts for the element relative to UAV $\mu$ from the original objective function $J$ in Eq. (\ref{eq:objective_function}). Furthermore, similarly to Eq. (\ref{eq:average_obstacle}), the objective function can be written solely as a function of the first $2$ moments of UAV $\mu$, hence solely as a function of the controls $\text{u}^\mu_{[0:T-1]}$. Equations (\ref{eq:eq2_user_f} - \ref{eq:eq6_user_f}) represent the same constraints on the control choices as presented in the \texttt{ATM CC-MPC} optimization, while Eqs. (\ref{eq:eq7_user_f} - \ref{eq:eq9_user_f}) represent the upperbound (as introduced by Lemma 2) on the chance constraints used by the \texttt{ATM CC-MPC} optimization to ensure that UAV $\mu$ is at least $\text{d}_\text{MIN}$ meters from any other UAV, with probability $1-\epsilon$. The deterministic upperbounds of the chance constraints Eqs. (\ref{eq:eq7_user_f} - \ref{eq:eq9_user_f}) are obtained as a function of $\mathbb{E}[f^{\mu,i}_k]$ and $\mathbb{E}[(f^{\mu,i}_k)^2]$, i.e, as a function of the first four moments of UAVs $\mu$ and $i$. While the first four moments of UAV $i$ are known at the moment of the optimizations (from old or new plans of UAV $i$, depending on whether $i>\mu$ or $i<\mu$, respectively), the four moments of the system state of UAV $\mu$ depend on the controlling variables $\text{u}^\mu_{[0:T-1]}$, which are therefore the only optimization variables in the \texttt{Distributed Det-MPC NG} optimization.

\section{PERFORMANCE EVALUATION}
\label{sec:performance_evaluation}

\subsection{Simulation Setup}
\label{subsec:sim_scenario}
To evaluate the capabilities of the introduced \texttt{Distributed Det-MPC NG} optimization, a simplified scenario with $M=4$ UAVs in a 3D flying environment is considered. In this scenario, the UAVs initiate their journeys from equidistant positions on a circle with a radius of $25$ meters. Each UAV is expected to land at the starting location of the opposite UAV. Specifically, the destination locations for each UAV are as follows: $[c_x^1, c_y^1, c_z^1] = [0, 25, 0]$; $[c_x^2, c_y^2, c_z^2] = [0, -25, 0]$; $[c_x^3, c_y^3, c_z^3] = [25, 0, 0]$; $[c_x^4, c_y^4, c_z^4] = [-25, 0, 0]$.

The safety constraint is designed to ensure a minimum distance of $\text{d}_\text{MIN}=10$m with probability $1-\epsilon=0.9$ or $1-\epsilon=0.99$, depending on the specific simulation. Therefore, the simulation scenario is crafted to maintain a substantial safety distance between UAVs within a relatively confined flying environment. The sampling time is $\Delta_s=0.1$s, while the planning horizon is $T=10$ time steps. UAVs start from a still position and are able to obtain an accurate estimation of their system state at each time step, i.e., without errors.  In all cases, the order applied to obtain UAVs' trajectories with the \texttt{Distributed Det-MPC NG} optimization adheres to the index cardinality of the UAVs. Further, limits on speed, direction, and altitude controls are as follow: $\text{u}^v_\text{MIN}=0$m/s; $\text{u}^v_\text{MAX}=10$m/s; $\text{u}^z_\text{MIN}=-10$m/s; $\text{u}^z_\text{MAX}=10$m/s; $\text{u}^\psi_\text{MIN}=-\pi$ rad/s; $\text{u}^\psi_\text{MAX}=\pi$ rad/s. To ensure smooth control selection, subsequent speed and altitude controls also respect the following maximum variations: $\Delta \text{u}^v=1$ m/s; and $\Delta \text{u}^z=1$ m/s. The weighting factor $w$ (balancing speed to destination and smoothness) was set to $0.1$. Motion model and disturbances are set as in Sec. \ref{sec:system_model}.

UAVs plan their trajectory over an $8$-second period, spanning $80$ time steps, in a distributed manner, as detailed in Sec. \ref{sec:user_optimum}. At the conclusion of this interval in all simulations, the UAVs consistently reach their intended destination, with a positional deviation of less than $1$ meter, and come to a complete stop. In our study, off-the-shelf non-linear optimization methods (i.e., interior-point) were employed using the Wolfram Mathematica software suite \cite{mathematica}. Multi-start setting was used, i.e., the solver was initialized from several initial solution. Each initialization run, on average, required $3$ minutes of computation time.

\subsection{\texttt{Distributed Det-MPC NG} Simulation Results}
\label{subsec:toy_explained}

Initially, in this section, the trajectories executed applying the \texttt{Distributed Det-MPC NG} optimization by the simulated UAVs are presented (Fig. \ref{fig:fulltrajectory}). To depict the average behavior of the proposed distributed approach, the UAVs' system states at each time step are computed using the average noise perturbation, i.e., using the first moment expression $m^1_k$. Figures \ref{fig:trajectory1} and \ref{fig:trajectory2} present two different viewpoints of the executed trajectories. First, it is shown that, as soon as UAV $1$ appears in the planning horizon of UAV $2$, UAV $2$ takes avoiding maneuvers, moving south-east. The chosen trajectory is optimal (UAV $2$ flying time is marginally larger than the one of UAV $1$), and enables UAV $2$ to maintain a close margin to the safety distance $\text{d}_\text{MIN}$ (as showed in Fig. \ref{fig:pairwisedistance}). Similarly, UAV $3$ avoids UAVs $1$ and $2$ by passing above their expected trajectories. Then, UAV $4$ cannot simply fly over or below the rest of the UAVs, thus it moves north-east, accurately selecting its controls so to respect $\text{d}_\text{MIN}$ from both UAVs $1$ and $3$ (as showcased by Fig. \ref{fig:pairwisedistance}). Interestingly, even if UAV $4$ takes the largest route to destination, the \texttt{Distributed Det-MPC NG} optimization efficiently and smoothly plans a safe route to destination which is only $0.4$ s longer than the shortest route. Finally, it must be stressed that the obtained multi-UAV planning has been obtained with non-Gaussian noise and that, even with the utilization of upper bounds for the original chance constraints, the trajectories executed by the UAVs consistently maintain a minimal distance from the safety threshold $\text{d}_\text{MIN}$.

\begin{figure*}[ht]
    \centering
    \includegraphics[width=0.92\textwidth]{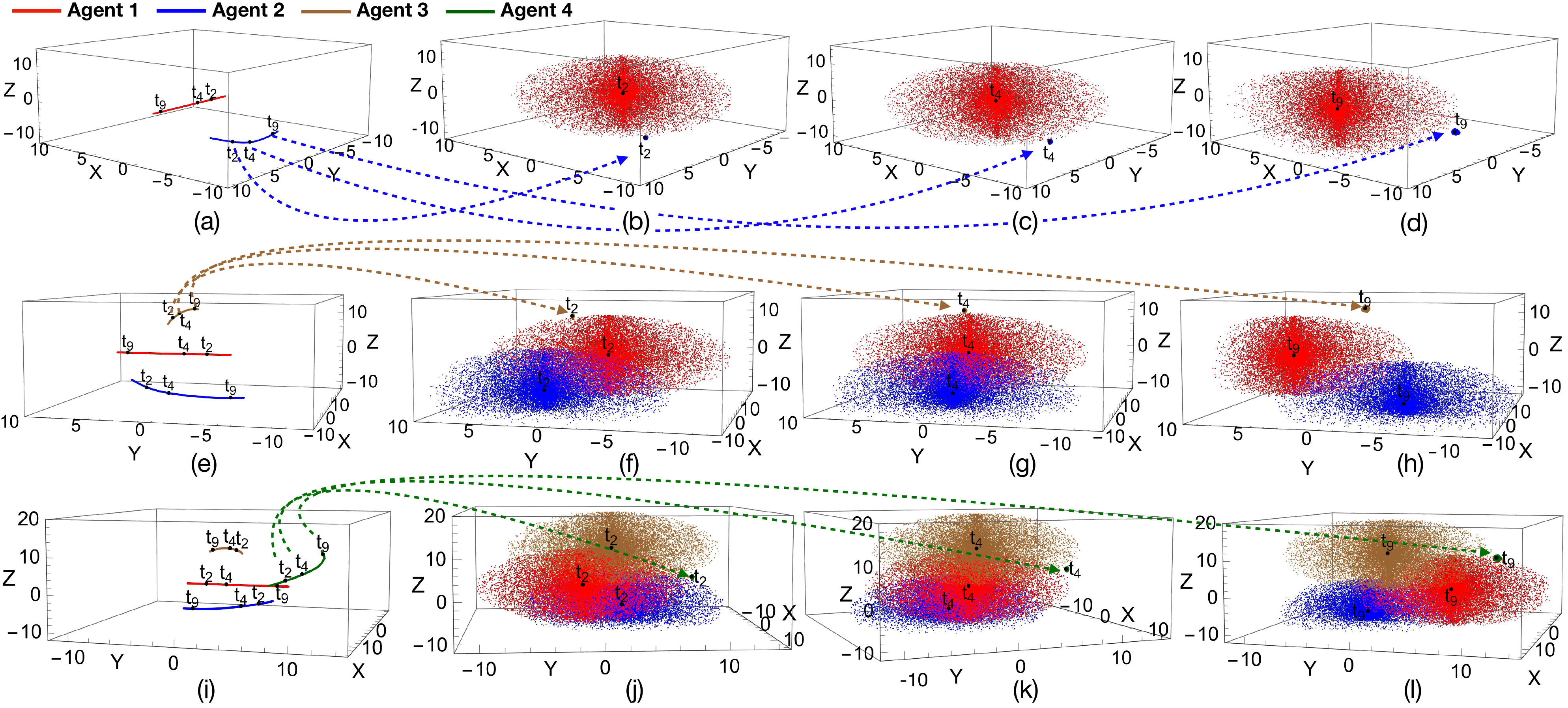}
    \vspace{-1mm}
    \caption{Trajectories planned at time $2.6$s respecting safety distance $\text{d}_\text{MIN}=10$ m with probability $1-\epsilon=0.9$ for non-Gaussian noise. Subplots (a$-$d) correspond to UAV $2$, (e$-$h) to UAV $3$, and (i$-$l) to UAV $4$. The second column shows a sampling of the safety distances, representing the volume in the environment that is closer than $\text{d}_\text{MIN}=10$ to another UAV, for the second time step of the planning horizon. The third and fourth columns depict the same for the fourth and ninth time steps of the planning horizon.}
    \label{fig:infeasibility_region}
\end{figure*}

\begin{figure}[ht]
    \centering
    \includegraphics[width=0.65\columnwidth]{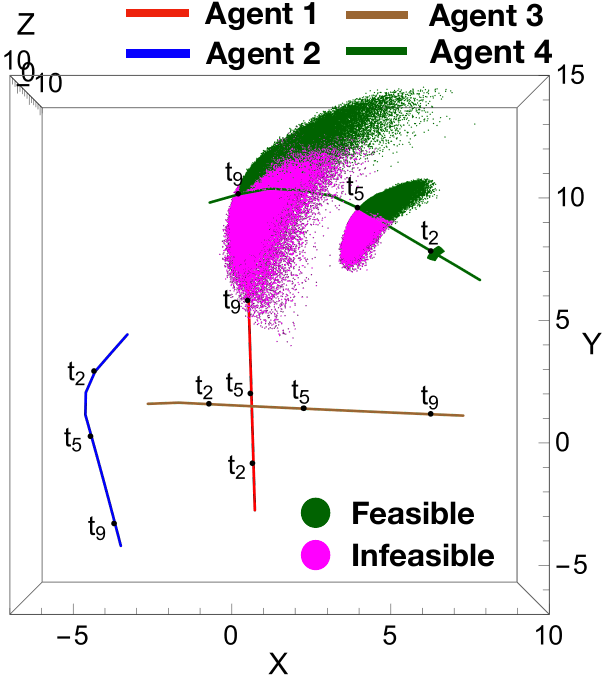}
    \vspace{-1mm}
     \caption{Trajectories planned by all UAVs at time $2.6$s. Green denotes locations reachable by UAV $4$ at the second, fourth, and ninth time steps, following safety-compliant control sequences. Pink indicates locations reachable using controls that do not adhere to safety constraints. UAV positions at selected time steps, planned using \texttt{Distributed Det-MPC NG}, are marked with black dots.}
     \label{fig:infeasible_control}
    %\vspace{-5mm}
\end{figure}

\begin{figure}[ht]
  \centering
  \begin{subfigure}[b]{0.65\columnwidth}
    \centering
    \includegraphics[width=\textwidth]{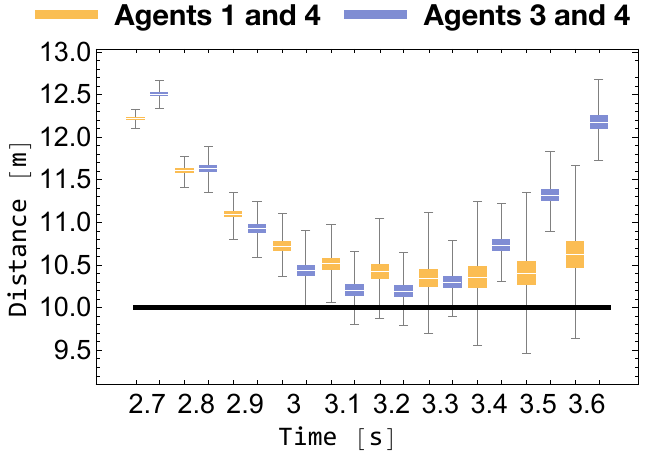}
    \caption{$\epsilon=0.1$}
    \label{fig:delta01}
  \end{subfigure}
  \\
  \begin{subfigure}[b]{0.65\columnwidth}
    \centering
    \includegraphics[width=\textwidth]{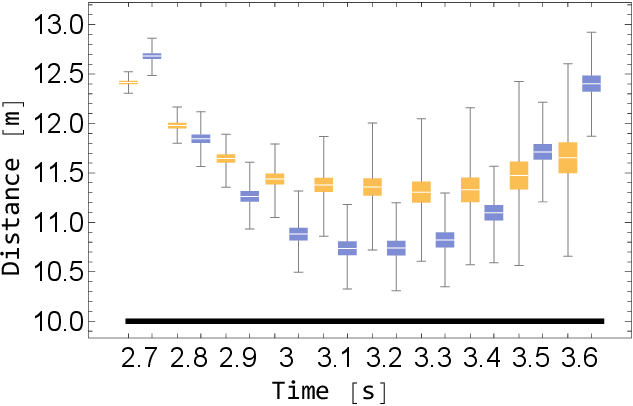}
    \caption{$\epsilon=0.01$}
    \label{fig:delta001}
  \end{subfigure}
  \caption{Pairwise distances between UAVs $1$ and $4$, and UAVs $3$ and $4$, calculated at 2.6 seconds using Monte Carlo simulations within the planning horizon. Boxplots for $\epsilon=0.1$ (a) and $\epsilon=0.01$ (b) show distances from the $25^{th}$ to the $75^{th}$ percentiles, with whiskers indicating the $1^{st}$ to the $99^{th}$ percentiles.}
  \label{fig:montecarlo_distance}
  \vspace{-3mm}
\end{figure}

Graphical representations of the planned trajectories at a specific time step, particularly at $2.6$ seconds, are presented in Fig. \ref{fig:infeasibility_region}. This time step was selected because it signifies the moment when all UAVs have reached the center of the flying environment and deftly avoid each other. It is important to note that, for visual simplicity, although only the new plans of the UAVs are shown, UAV planning considers chance constraints for all other UAVs in the shared flying environment. For instance, the planned trajectory for UAV 2 (first row of Fig. \ref{fig:infeasibility_region}) takes into account the old trajectories of UAVs 3 and 4, even though they are not shown in the figure. In Fig. \ref{fig:infeasibility_region}, by employing the controls derived from \texttt{Distributed Det-MPC NG} optimization, and recursively computing UAV locations at the second, fourth, and ninth time steps using Eq. (\ref{eq:agent_dynamics_augmented}) and $1000$ different control disturbance realizations, particles representing the UAV locations for the planning horizon are obtained. Subsequently, Fig. \ref{fig:infeasibility_region} displays $500$ points for each particle, selected within a sphere of radius $\text{d}_\text{MIN}$ centered on the respective particle, in red, blue, and brown, for UAVs $1$, $2$, and $3$, respectively. Hence, Fig. \ref{fig:infeasibility_region} illustrates visually the region in the flying environment that is closer than $\text{d}_\text{MIN}$ to any other UAV, guiding the trajectory decisions made by the \texttt{Distributed Det-MPC NG} optimization. The first row of Fig. \ref{fig:infeasibility_region} illustrates UAV $2$ shaping its trajectory at $2.6$ seconds, avoiding the approaching UAV $1$ by moving southeast. UAV $2$ consistently steers clear of the region closer than $\text{d}_\text{MIN}$ to UAV $1$ at all time steps while minimizing its distance from the destination. The second and third rows illustrate the regions avoided by UAVs $3$ and $4$, respectively, to obtain safe trajectories. Notably, UAV $4$ strategically selects its trajectory, capitalizing on the movements of other UAVs to navigate through a narrow gap between regions closer than $\text{d}_\text{MIN}$ to UAVs $1$ and $3$.

Fig. \ref{fig:infeasible_control} offers an alternative view of UAV $4$'s trajectory at the same time step, displaying top-view perspectives of all UAV trajectories. Furthermore, for this figure, $10000$ controls for UAV $4$ were randomly selected around the optimal controls given by the \texttt{Distributed Det-MPC NG} optimization (differences between subsequent controls in the selected sequences only deviate at most $20\%$ from what happens in the optimal controls).
The green (pink) markers on the figure represent locations, as obtained in Eq. (\ref{eq:first_moment}), reached by UAV $4$ adhering to (violating) the safety constraints in Eq. (\ref{eq:VP_1}) at the second, fifth, and ninth time steps within the planning horizon. Notably, the figure illustrates that to maximize UAV $4$'s forward motion while respecting safety constraints, \texttt{Distributed Det-MPC NG} navigates a complex optimization challenge.

In Fig. \ref{fig:montecarlo_distance}, the effectiveness of the deterministic safety constraints approximating the safety chance constraint in Eq. (\ref{eq:safe_distance}) using concentration inequalities is demonstrated. Fig. \ref{fig:delta01} displays the pairwise distance distribution between UAVs $1$ and $4$, and UAVs $3$ and $4$ throughout the entire planning horizon, assuming $\epsilon=0.1$ using the same particles generated in Fig. \ref{fig:infeasibility_region}. Similarly, Fig. \ref{fig:delta001} shows the distribution when the controls are selected for $\epsilon=0.01$ for the same optimization in Fig. \ref{fig:delta01}. Both cases maintain the desired safety level with a slight margin. With $\epsilon=0.1$ (indicating a $10\%$ probability of violating the safety distance $\text{d}_{\text{MIN}}$), $3.87\%$ and $4.15\%$ of particles violate $\text{d}_{\text{MIN}}$ between UAVs 1 and 4 and between UAVs 3 and 4 (at the fifth and ninth time step, respectively). With $\epsilon=0.01$, the minimum distance between particles is $10.21$ meters for UAVs $1$ and $4$ and $10.09$ meters for UAVs $3$ and $4$, i.e, just above $\text{d}_{\text{MIN}}$. Notably, even in presence of non-Gaussian noise, our solution is not overly conservative.

\section{CONCLUSIONS}
\label{sec:conclusion}
In contrast to existing solutions, this work presents a probabilistically robust distributed controller specifically designed for UAVs operating in challenging non-Gaussian environments under non-linear dynamics. Employing advanced techniques like mixed-trigonometric-polynomial exact moment propagation, non-Gaussian chance constraints are transformed into their deterministic approximated counterparts. Our approach integrates these constraints into a MPC framework, guiding UAVs in selecting controls and planning trajectories. The performance evaluation shows that our methodology successfully satisfies safety while maximizing the desired objective, as shown in different context. Importantly, this result is obtained while ensuring safety without being unnecessary conservative, even in the presence of non-Gaussian perturbations. Looking ahead, an investigation into the feasibility of scaling the approach to real-world settings and exploring lightweight implementations of the centralized controller introduced in this work will be conducted.

\bibliographystyle{IEEEtran}
\bibliography{IEEEabrv,main}

\end{document}